\documentclass[12pt,a4paper, DIV18]{scrartcl}
\usepackage[utf8]{inputenc}
\usepackage{graphicx}
\usepackage{stfloats}
\usepackage{amssymb, amsthm}
\usepackage[cmex10]{amsmath}
\usepackage{color}
\usepackage{array}
\usepackage[pagebackref=true]{hyperref}
\hypersetup{colorlinks=false,pdfborderstyle={/S/U/W 1},pdfborder=0 0 1}

\usepackage{cite}
\usepackage{relsize}
\usepackage[ruled,vlined,titlenumbered]{algorithm2e}
\usepackage{tabularx}
\usepackage{arydshln}
\usepackage{tikz}[2010/10/13 v2.10]
\usepackage{pgfplots}
\usepackage{textcomp}
\usepackage{wasysym}

\usetikzlibrary{chains,positioning,arrows,calc}
\usepackage{nicefrac}
\usepgfplotslibrary{external}
\tikzsetexternalprefix{tikzfig/}

\pgfplotscreateplotcyclelist{mylist}{red,blue,black,yellow,brown}

\let\myLambda\Lambda 
\renewcommand{\Lambda}{\mathrm{\myLambda}}
\let\myOmega\Omega 
\renewcommand{\Omega}{\mathrm{\myOmega}}
\let\myDelta\Delta
\renewcommand{\Delta}{\mathrm{\myDelta}}

\newcommand{\LCNOM}[1]{\dfrac{\sum\limits_{j \in \mathcal{Z}} \Big( a_j \prod \limits_{\substack{\ell \in \mathcal{Z}\\ \ell \neq j}} (1-x \alpha^{#1} \beta^\ell ) \Big)}{ \prod\limits_{j \in \mathcal{Z}} (1-x \alpha^{#1} \beta^j)}}

\newcolumntype{Z}{>{\centering\arraybackslash}X}

\usepackage{ifthen}
\newcommand{\F}[1]{\ensuremath{\mathbb{F}_{#1}}}
\newcommand{\Fq}{\F{q}}
\newcommand{\Fqn}[2]{\ensuremath{\mathbb{F}_{#1}^{#2}}}
\newcommand{\Fx}[1]{\ensuremath{\F{#1}[x]}}
\newcommand{\Fxq}{\Fx{q}}
\newcommand{\Z}[1]{\ensuremath{\mathbb{Z}_{#1}}} 
\newtheorem{corollary}{Corollary}
\newtheorem{definition}{Definition}
\newtheorem{theorem}{Theorem}
\newtheorem{lemma}{Lemma}
\newtheorem{proposition}{Proposition}
\newtheorem{example}{Example}
\newtheorem{remark}{Remark}
\DeclareMathOperator{\defi}{def}
\newcommand{\defeq}{\overset{\defi}{=}}
\newcommand{\defequiv}{\overset{\defi}{\equiv}}

\newcommand{\defset}[1]{\ensuremath{D_\mathcal{#1}}}
\newcommand{\refeq}[1]{(\ref{#1})}

\newcommand{\RS}[4]{\ensuremath{\mathcal{RS}(#1;#2,#3;#4)}}
\newcommand{\RSa}{\ensuremath{\mathcal{RS}}}

\newcommand{\CYC}[4]{\ensuremath{\mathcal{C}(#1;#2,#3,#4)}}
\newcommand{\CYCa}{\ensuremath{\mathcal{C}}}
\newcommand{\LC}[4]{\ensuremath{\mathcal{L}(#1;#2,#3,#4)}}
\newcommand{\LCa}{\ensuremath{\mathcal{L}}}
\newcommand{\LCn}{\ensuremath{n_\mathcal{\ell}}}
\newcommand{\LCk}{\ensuremath{k_\mathcal{\ell}}}
\newcommand{\LCd}{\ensuremath{d_\mathcal{\ell}}}
\newcommand{\LCq}{\ensuremath{q_\mathcal{\ell}}}
\newcommand{\LCs}{\ensuremath{s_\mathcal{\ell}}}
\newcommand{\LCextensionorder}{\ensuremath{u}}
\newcommand{\LCconst}{\ensuremath{e}}
\newcommand{\mulo}{\ensuremath{\mu}}
\newcommand{\SPC}[1]{\ensuremath{\mathcal{P}(#1,#1-1,2)}}
\newcommand{\SPCa}{\ensuremath{\mathcal{P}}}
\newcommand{\HTconsta}{\ensuremath{b_1}}
\newcommand{\HTconstb}{\ensuremath{b_2}}
\newcommand{\muHT}{\ensuremath{d_0}}
\newcommand{\nuHT}{\ensuremath{\nu}}
\newcommand{\HTa}{\ensuremath{m_1}}
\newcommand{\HTb}{\ensuremath{m_2}}
\newcommand{\COS}[2]{\ensuremath{M_{#1}^{(#2)}}}
\renewcommand{\vec}[1]{\ensuremath{\mathbf{#1}}}
\newcommand{\vecenum}[1]{\ensuremath{(#1_0 \ #1_1 \ \dots \ #1_{n-1})}}

\begin{document}

\title{A New Bound on the Minimum Distance of Cyclic Codes Using Small-Minimum-Distance Cyclic Codes}


\author{Alexander Zeh and Sergey Bezzateev\footnote{Alexander Zeh is with the Institute of Communications Engineering, University of Ulm, Ulm, Germany and INRIA Saclay--\^{I}le-de-France, École Polytechnique ParisTech, Palaiseau Cedex, France. Sergey Bezzateev is with the Saint Petersburg State University of Airspace Instrumentation, St. Petersburg, Russia, Email: \href{mailto:alexander.zeh@uni-ulm.de}{\texttt{alexander.zeh@uni-ulm.de}}, \href{mailto:bsv@aanet.ru}{\texttt{bsv@aanet.ru}}. The material in this contribution was presented in part to the IEEE International Symposium on Information Theory (ISIT 2012) in Boston, USA \cite{ZehBezzateev_DescribingACyclicCodeByAnotherCyclicCode_2012}.}}


\maketitle

\subsubsection*{Abstract}
A new bound on the minimum distance of $q$-ary cyclic codes is proposed. It is based on the description by another
cyclic code with small minimum distance. The connection to the BCH bound and the Hartmann--Tzeng (HT) bound is formulated explicitly. We show that for many cases our approach improves the HT bound. Furthermore, we refine our bound for several families of cyclic codes.

We define syndromes and formulate a Key Equation that allows an efficient decoding up to our bound with the Extended Euclidean Algorithm.
It turns out that lowest-code-rate cyclic codes with small minimum distances are useful for our approach. Therefore, we give a sufficient condition for binary cyclic codes of arbitrary length to have minimum distance two or three and lowest code-rate.\\[0.1cm]

\textbf{Keywords: BCH Bound - Bound on the Minimum Distance - Cyclic Code - Decoding - Hartmann--Tzeng Bound}\\[0.1cm]

\textbf{Mathematics Subject Classification: 94A24 - 94A55 - 94B15 - 94B35}

\section{Introduction} \label{sec_intro}
In this paper, we introduce a technique that uses an $(\LCn,\LCk)$ $\LCq$-ary cyclic code $\LCa$ with minimum distance $\LCd$ to bound the minimum distance $d$ of another $(n,k)$ $q$-ary cyclic code $\CYCa$.
The descriptive cyclic code $\LCa$ is called non-zero-locator code. It turns out that the non-zero-locator code gives a good lower bound $d^{\ast}$ on the minimum distance $d$ of the described cyclic code $\CYCa$ if the code-rate $\LCk/\LCn$ of $\LCa$ is low and its minimum distance $\LCd$ is relatively small.

The algebraic relation between the cyclic non-zero-locator code $\LCa$ and the cyclic code $\CYCa$ provides the formulation
of syndromes and a Key Equation that allows an efficient decoding up to $\lfloor (d^{\ast}-1)/2 \rfloor$ errors with the Extended Euclidean Algorithm (EEA).

We give an explicit relation of $d^{\ast}$ to the BCH bound~\cite{Hocquenghem_1959,Bose_RayChaudhuri_1960} and its generalization: the Hartmann--Tzeng (HT) bound~\cite{Hartmann_DecodingBeyondBCHbound_1972,Hartmann_GeneralizationsofBCHbound_1972,Hartmann_DecodingBeyondBCHBound_1974}. In many cases our bound is better than the HT bound, although our approach is not a generalization of the HT bound as the Roos bound~\cite{Roos_GeneralBCHBound_1982,Roos_BoundforCyclicCodes_1983} and the bound of van Lint and Wilson~\cite{vanLint_OnTheMinimumDistance_1986} are.

In our previous work~\cite{ZehWachterBezza_EfficientDecodingOfSomeClassesArxiv_2011} we associated rational functions with a subset of the defining set of a given cyclic code $\CYCa$. This can be seen as a special case of the presented approach.
The main advantage of this contribution is that we can express the bound on the minimum distance of a given cyclic code $\CYCa$ in terms of properties of the associated cyclic non-zero-locator code $\LCa$.

This paper is organized as follows. In Section~\ref{sec_preliminaries}, we give necessary preliminaries of cyclic codes, the HT bound and recall the definition of cyclic Reed--Solomon (RS) codes, which we use later as non-zero-locator code.
The concept of the non-zero-locator code is introduced in Section~\ref{sec_loccode} and the main theorem on the minimum distance is proven. 
The connection to the Hartmann--Tzeng bound is given in Section~\ref{sec_LocatorCodes}. Furthermore, several families of cyclic codes are identified. We give sufficient conditions for binary cyclic codes with minimum distance two and three and lowest code-rate in Section~\ref{sec_badcyclic}. A generalized syndrome definition, Key Equation and Forney's formula are given in Section~\ref{sec_decoding}. Section~\ref{sec_conclusion} concludes this contribution. 

\section{Preliminaries} \label{sec_preliminaries}
Let $q$ be a power of a prime and let \F{q} denote the finite field of order $q$ and \Fxq{} the set of all univariate polynomials with coefficients in \Fq{} and indeterminate $x$.
A $q$-ary cyclic code over \Fq{} of length $n$, dimension $k$ and minimum distance $d$ is
denoted by $\CYC{q}{n}{k}{d} \subset \Fqn{q}{n}$ and it is an ideal in the ring $\Fxq{}/(x^n-1)$ generated by $g(x)$. A codeword $\vec{c} =\vecenum{c} \in \CYCa$
is associated with a polynomial $c(x) = \sum_{i=0}^{n-1} c_i x^i \in \Fxq{}$, where
$g(x)$ divides $c(x)$. We assume that $x^n-1$ has $n$ different roots. Let $\F{q^s}$ be an extension field of \Fq{} and let $\alpha \in \F{q^s}$ be a primitive $n$th root of unity.
The cyclotomic coset \COS{r}{n} modulo $n$ over \F{q} is denoted by:
\begin{equation*} \label{eq_cyclotomiccoset}
 \COS{r}{n} = \lbrace rq^j \bmod n \, \vert \, j = 0,1,\dots,n_r-1 \rbrace,
\end{equation*}
where $n_r$ is the smallest integer such that $rq^{n_r} \equiv
r \mod n$. It is well--known that the minimal polynomial $\COS{r}{n}(x) \in \Fxq{}$ of the element $\alpha^r$
is given by:
\begin{equation*} \label{eq_minpoly}
 \COS{r}{n}(x) = \prod_{i \in \COS{r}{n}} (x-\alpha^i).
\end{equation*}
The defining set $\defset{\CYCa}$ of a $q$-ary cyclic code
\CYC{q}{n}{k}{d} is the set of zeros of the generator
polynomial $g(x) \in \Fxq$ and can be partitioned into $m$ cyclotomic
cosets:
\begin{equation*} \label{eq_definingset}
\begin{split}
\defset{\CYCa} & =  \{0 \leq i \leq n-1 \, | \, g(\alpha^i)=0 \} = \COS{r_1}{n} \cup \COS{r_2}{n} \cup \dots \cup \COS{r_{m}}{n}.
\end{split}
\end{equation*}
Hence, the generator polynomial $g(x)$ of degree $n-k$ of \CYC{q}{n}{k}{d} is
\begin{equation*}
 g(x) = \prod_{i=1}^{m} \COS{r_i}{n}(x).
\end{equation*}
Let us recall a well-known bound on the minimum distance of cyclic codes.
\begin{theorem}[Hartmann--Tzeng (HT) Bound,~\cite{Hartmann_GeneralizationsofBCHbound_1972}] \label{theo_HT}
Let a $q$-ary cyclic code $\CYC{q}{n}{k}{d}$ with defining set
$D_{\mathcal{C}}$ be given.
Suppose there exist the integers $\HTconsta$, $\HTa$ and $\HTb$ with $\gcd(n,\HTa) = 1$ and $\gcd(n,\HTb) = 1$ such that
\begin{equation*} \label{eq_HTBound}
 \{ \HTconsta + i_1\HTa+i_2\HTb \mid 0 \leq i_1 \leq \muHT - 2,\; 0 \leq i_2 \leq \nuHT \} \subseteq \defset{\CYCa}.
\end{equation*}
Then $d \geq \muHT+\nuHT$.
\end{theorem}
Note that for $\nuHT =0$ the HT bound becomes the BCH bound~\cite{Hocquenghem_1959, Bose_RayChaudhuri_1960}.
Further generalizations were proposed by Roos~\cite{Roos_GeneralBCHBound_1982,Roos_BoundforCyclicCodes_1983} and van Lint and Wilson~\cite{vanLint_OnTheMinimumDistance_1986}.
Decoding up to the HT bound and to some particular cases of the Roos bound was formulated by Feng and Tzeng~\cite[Section VI]{Feng-Tzeng:IEEE_IT1991}.

We consider cyclic Reed--Solomon (RS) codes~\cite{ReedSolomon_PolynomialCodesOverCertainFiniteFields_1960} for our approach and therefore recapitulate their definition in the following.
\begin{definition}[Cyclic Reed--Solomon Code] \label{def_CyclicRS}
Let $n$ be an integer dividing $q-1$ and let $\alpha$ denote an element of multiplicative order $n$ in \Fq{}.
Let $\delta$ be an integer.
Furthermore, let the generator polynomial $g_{\delta}(x) \in \Fxq$ be defined as:
\begin{equation*} \label{eq_GenPolyRS}
g_{\delta}(x) = \prod\limits_{i=\delta}^{\delta+n-k-1} (x-\alpha^i).
\end{equation*}
Then a cyclic Reed--Solomon code over $\Fq{}$ of length $n | (q-1)$ and dimension $k$, denoted by $\RS{q}{n}{k}{\delta}$, is defined by:
\begin{equation} \label{eq_DefGenPolyRSCode}
\RS{q}{n}{k}{\delta} = \{ m(x) g_{\delta}(x) :  \deg m(x) < k \}.
\end{equation}
\end{definition}
RS codes are maximum distance separable codes and their minimum distance $d$ is $d=n-k+1$.

\section{The Non-Zero-Locator Code} \label{sec_loccode}
We relate another cyclic code --- the so-called non-zero-locator code $\LCa{}$ --- to a given cyclic code $\CYCa{}$.
In the following, we connect a infinite sequence of an evaluated polynomial $c(x) \in \CYCa$ to a sum of fractions. This allows to draw the relation to our previous approach~\cite{ZehWachterBezza_EfficientDecodingOfSomeClassesArxiv_2011}. Furthermore, we can use familiar properties of cyclic codes rather than abstract properties of rational functions. The obtained bound can be expressed in terms of parameters of the associated non-zero-locator code $\LCa$.

Let $c(x)$ be a codeword of a given $q$-ary cyclic code \CYC{q}{n}{k}{d} and let $\mathcal{Y}$ denote the set of indexes of non-zero coefficients of $c(x)$
\begin{equation*}
c(x) = \sum\limits_{i \in \mathcal{Y}} c_{i} x^{i}.
\end{equation*}
Let $\alpha \in \F{q^s}$ be an element of order $n$. Then we have the following
relation for all $c(x) \in \CYC{q}{n}{k}{d}$:
\begin{align} \label{eq_PowerSeriesCyclicCode}
 \sum\limits_{j=0}^{\infty} c(\alpha^j) x^j & = \sum\limits_{j=0}^{\infty} \sum\limits_{i \in \mathcal{Y}} c_i \alpha^{ji}x^j \nonumber \\
 & = \sum\limits_{j=0}^{\infty} \sum\limits_{i \in \mathcal{Y}} c_i (\alpha^{i}x)^j \nonumber \\
  & = \sum\limits_{i \in \mathcal{Y}} \sum\limits_{j=0}^{\infty} c_i (\alpha^{i}x)^j \nonumber \\
 & = \sum\limits_{i \in \mathcal{Y}} \frac{c_i}{1-x \alpha^i}.
\end{align}
Now, we can define the non-zero-locator code.
\begin{definition}[Non-Zero-Locator Code] \label{def_locatorcode}
Let a $q$-ary cyclic code \CYC{q}{n}{k}{d} be given. Let \F{q^s} contain the $n$th roots of unity. Let $\gcd(n,\LCn) = 1$ and let 
$\F{\LCq} = \F{q^\LCextensionorder}$ be an extension field of \F{q}. Let \F{\LCq^{\LCs}} contain the $\LCn$th roots of unity. Let $\alpha \in \F{q^s}$ be an element of order $n$ and let $\beta \in \F{\LCq^{\LCs}}$ be an element of order $\LCn$.

Then \LC{\LCq}{\LCn}{\LCk}{\LCd} is a non-zero-locator code of $\CYCa$ if there exists a $\mu \geq 2$ and an integer $\LCconst$, such that $\forall \, a(x) \in \LCa$ and $\forall \, c(x) \in \CYCa$:
\begin{equation} \label{eq_locatorandcodeword}
\sum_{j=0}^{\infty} c(\alpha^{j+\LCconst}) a(\beta^j) x^j \equiv 0 \bmod x^{\mulo-1},
\end{equation}
holds.
\end{definition}
\begin{remark}
Let $r$ denote the least common multiple of $s$ and $\LCextensionorder \cdot \LCs$ and let $\gamma$ be a primitive element
in $\F{q^{r}}$. Then $\gamma^{(q^{r}-1)/n}$ and $\gamma^{(q^{r}-1)/\LCn}$ are elements of order $n$ and $\LCn$.
\end{remark}
Before we prove the main theorem on the minimum distance $d$ of the given cyclic code $\CYCa$, we describe Definition~\ref{def_locatorcode}.
We search the ``longest'' sequence
\begin{equation*}
c(\alpha^\LCconst) a(\beta^0), c(\alpha^{\LCconst+1}) a(\beta^1), \dots , c(\alpha^{\LCconst+\mulo-2}) a(\beta^{\mulo-2}),
\end{equation*}
that results in a zero-sequence of length $\mulo-1$, i.e., the product of the evaluated codeword $a(\beta^j)$ of the non-zero-locator code $\LCa$ and the evaluated codeword $c(\alpha^{j+\LCconst})$ of $\CYCa$ gives zero for all $j=0,\dots,\mulo-2$. Let us study the following example of a binary cyclic code.
\begin{example}[Binary Code of length $n=21$~\cite{Roos_BoundforCyclicCodes_1983, vanLint_OnTheMinimumDistance_1986}]
\label{ex_Binary21a} Let the binary cyclic code \CYC{2}{21}{7}{8} with generator polynomial $g(x)$
\begin{equation*}
g(x) =  \COS{1}{21}(x) \cdot  \COS{3}{21}(x) \cdot \COS{7}{21}(x) \cdot \COS{9}{21}(x)
\end{equation*}
be given. Let $\alpha \in \F{2^6}$ denote an element of order $21$. \\
The defining set $\defset{\CYCa} = \COS{1}{21} \cup \COS{3}{21} \cup \COS{7}{21} \cup \COS{9}{21}$ of \CYC{2}{21}{7}{8} is
\begin{equation*}
\begin{split}
\defset{\CYCa} = \{ 1,2,3,4,\square,6,7,8,9,\square,11,12,\square,14,15,16,\square,18 \},
 \end{split}
\end{equation*}
where the symbol $\square$ marks the indexes where $g(\alpha^i) \neq 0$.

We associate a single parity check code of length $\LCn = 5$, dimension $\LCk=4$ and minimum distance $\LCd = 2$ over $\F{2}$ as non-zero-locator code for \CYC{2}{21}{7}{8} according to Definition~\ref{def_locatorcode}. Let $\beta \in \F{2^4}$ be an element of order $5$ and let $g(x) = x-1$ be the generator polynomial of \LCa{}.
The defining sets $\defset{\CYCa}$ of \CYC{2}{21}{7}{8} and $\defset{\LCa}$ of \LC{2}{5}{4}{2} are listed in Table~\ref{tab_Ex21RS}. The corresponding product gives the a zero-sequence of length $\mulo-1 = 13$ for $\LCconst=0$.
A codeword $a(x) \in \LC{2}{5}{4}{2}$ ``fills'' the missing zeros of \CYC{2}{21}{7}{8} at position $0$, $5$ and $10$ in the interval $\left[0,12\right]$.

\begin{table}[htbp]
\centering
\caption{Defining sets $\defset{\CYCa}$ and $\defset{\LCa}$ of the binary cyclic code \CYC{2}{21}{7}{8} of Example~\ref{ex_Binary21a} and its non-zero-locator code \LC{2}{5}{4}{2} in the interval $\left[0,12\right]$.}
\label{tab_Ex21RS}
\begin{tabularx}{.95\columnwidth}{l|ZZZZZ;{2pt/2pt}ZZZZZ;{2pt/2pt}ZZZ}
$\defset{\CYCa}$ & $\square$ & 1 & 2 & 3 & 4 & $\square$ & 6 & 7 & 8 & 9 & $\square$ & 11 & 12 \\ 
$\defset{\LCa}$ & 0 & $\square$ & $\square$ & $\square$ & $\square$ & 0 & $\square$ & $\square$ & $\square$ & $\square$ & 0 & $\square$ & $\square$ \\ 
\end{tabularx}
\end{table}
\end{example}
We require a zero $\beta^j$ of the generator polynomial of the  non-zero-locator code $\LCa$ at the position $j$ where the generator polynomial of the given cyclic code $\CYCa$ has no zero.

Furthermore, we require $\gcd(n,\LCn) = 1$ to guarantee that
\begin{equation*} \label{eq_coprime}
 \gcd \Big( \prod_{m \in \mathcal{Z}}  (1-x\alpha^i\beta^m), \prod_{m \in \mathcal{Z}} (1-x\alpha^j\beta^m) \Big) = 1 \quad \forall i \ \text{and} \ \forall j \neq i,
\end{equation*}
which we use for the degree calculation in the following. For the proof we refer to Lemma~\ref{lemma_appendixcoprimality} in the Appendix.

We rewrite~\refeq{eq_locatorandcodeword} of Definition~\ref{def_locatorcode} more explicitly. With $c(x) = \sum_{i \in \mathcal{Y}} c_{i} x^{i}$ and
 $a(x) = \sum_{j \in \mathcal{Z}} a_{j} x^{j}$, we obtain:
\begin{align*} 
\sum_{j=0}^{\infty}  c(\alpha^{j+\LCconst}) a(\beta^j) x^j & = \sum\limits_{j=0}^{\infty} \sum\limits_{i \in \mathcal{Y}} c_i \alpha^{i(j+\LCconst)} a(\beta^j)x^j \\
& =  \sum\limits_{i \in \mathcal{Y}} c_i \alpha^{i\LCconst} \sum\limits_{j=0}^{\infty} \alpha^{ij} a(\beta^j)x^j.
\end{align*}
Using~\refeq{eq_PowerSeriesCyclicCode} for the codeword $a(x)$ of the associated non-zero-locator code leads to:
\begin{align} \label{eq_locatorandcodeword3}
\sum\limits_{i \in \mathcal{Y}} c_i \alpha^{i\LCconst} \sum\limits_{j=0}^{\infty} \alpha^{ij} a(\beta^j)x^j & =  \sum\limits_{i \in \mathcal{Y}} c_i \alpha^{i\LCconst} \sum\limits_{j \in \mathcal{Z}} \frac{a_j}{1-x\alpha^i \beta^j} \nonumber \\
& = \sum \limits_{i \in \mathcal{Y}} c_i \alpha^{i\LCconst} \LCNOM{i}.
\end{align}
Finally using~\refeq{eq_locatorandcodeword3} we can rewrite~\refeq{eq_locatorandcodeword} of Definition~\ref{def_locatorcode} in the following form:
\begin{align} \label{eq_locatorandcodeword5}
\dfrac{\sum\limits_{i \in \mathcal{Y}} \Big( c_i \alpha^{i\LCconst} \sum\limits_{j \in \mathcal{Z}} \Big(a_j \prod\limits_{\substack{\ell \in \mathcal{Z}\\ \ell \neq j}} (1-x\alpha^i \beta^{\ell}) \Big) \prod\limits_{\substack{m \in \mathcal{Y}\\ m \neq i}} \prod\limits_{s \in \mathcal{Z}} (1-x\alpha^m \beta^s) \Big)  }{\prod\limits_{i \in \mathcal{Y}} \big( \prod\limits_{j \in \mathcal{Z}} (1-x\alpha^i \beta^j) \big)} & \equiv 0 \bmod x^{\mulo-1},
\end{align}
where the degree of the denominator is $| \mathcal{Y} | \cdot | \mathcal{Z} |$. The degree of the numerator is smaller than or equal to $(|\mathcal{Y}|-1) \cdot | \mathcal{Z}| +  | \mathcal{Z}| -1 = |\mathcal{Y}| \cdot |\mathcal{Z}| -1$.

This leads to the following theorem on the minimum distance of a cyclic code $\CYCa$.
\begin{theorem}[Minimum Distance] \label{theo_minimumdistance}
Let a $q$-ary cyclic code \CYC{q}{n}{k}{d} and its associated non-zero-locator code \LC{\LCq}{\LCn}{\LCk}{\LCd} with $\gcd(n,\LCn) = 1$ and the integer $\mu$ be given as in Definition~\ref{def_locatorcode}. Then the minimum
distance $d$ of \CYC{q}{n}{k}{d} satisfies the following inequality:
\begin{equation} \label{inequality_for_d}
d \geq d^{\ast} \defeq \left \lceil \frac{\mulo}{\LCd} \right \rceil.
\end{equation}
\end{theorem}
\begin{proof}
For a codeword $c(x) \in \CYC{q}{n}{k}{d}$ of weight $d$ and a codeword $a(x) \in \LC{\LCq}{\LCn}{\LCk}{\LCd}$ of weight $\LCd$, the degree of the denominator in~\refeq{eq_locatorandcodeword5} is $d \cdot \LCd$. The numerator has degree at most $d \cdot \LCd-1$, and has to be greater than or equal to $\mulo-1$.
\end{proof}

\begin{example}[Binary Code of length $n=21$] 
Let us again consider the binary code \CYC{2}{21}{7}{8} of Example~\ref{ex_Binary21a}. We have $\mu -1 = 13$ according to Theorem~\ref{theo_minimumdistance}, so $d^{\ast} = \lceil 14/2 \rceil = 7$.

The HT bound (Theorem~\ref{theo_HT}) gives also $d \geq 6$ (with parameters $\HTconsta=1$, $\HTa=5$,  $\HTb=1$, $\muHT=5$ and $\nuHT=1$). The Roos bound gives $d \geq 8$~\cite[Example 1]{vanLint_OnTheMinimumDistance_1986}, which is the actual minimum distance of $\CYC{2}{21}{7}{8}$.
\end{example}
The optimal non-zero-locator code \LCa{} for a given cyclic code gives a zero sequence
\begin{equation*}
c(\alpha^{\LCconst}) a(\beta^0) , c(\alpha^{\LCconst+1}) a(\beta^{1}) , \dots , c(\alpha^{\LCconst+\mulo-2}) a(\beta^{\mulo-2})
\end{equation*}
of length $\mulo-1$ as in Definition~\ref{def_locatorcode}, such that $d^{\ast}$ of~\refeq{inequality_for_d} is maximized.

\section{Comparison to Known Bounds} \label{sec_LocatorCodes}
\subsection{The Hartmann--Tzeng Bound}
We restate the HT bound as given in Theorem~\ref{theo_HT} to draw a connection to the bound given in Theorem~\ref{theo_minimumdistance}. We multiply with the inverse of $\HTa$ or $\HTb$ modulo $n$, such that:
\begin{equation} \label{eq_specialHT}
 m>\nuHT+1 \quad \text{for} \quad  \{ \HTconstb + i_{1}m+i_{2} : 0 \leq i_1 \leq \muHT-2, 0 \leq i_2 \leq \nu \} \subseteq D_{\mathcal{C}}
\end{equation}
with $\gcd(n,m)=1$ for a given code \CYC{q}{n}{k}{d} holds.

Throughout this section, we refer to this representation of the HT bound.
In the following subsection, we consider a single parity check code as non-zero-locator code and draw the connection to a particular case of the HT bound.
The general case of~\refeq{eq_specialHT} is considered in Subsection~\ref{subsec_RS}, where we use RS codes as non-zero-locator codes.

Some families of cyclic codes are identified in Subsection~\ref{subsec_familiesOfCode}.

\subsection{Single Parity Check Code as Non-Zero-Locator Code} \label{subsec_ParityCheck}
Let \SPC{\LCn} denote a cyclic single parity check code of length $\LCn$, dimension $\LCn-1$ and minimum distance $2$ over an extension field $\F{\LCq}$ of \F{q}. Let $\beta$ be a primitive $\LCn$th root of unity in an extension field of $\F{\LCq}$. The generator polynomial $g(x)$ of \SPC{\LCn} is
\begin{equation*} \label{eq_genpoly_singleparity}
g(x) = x-1.
\end{equation*}
Furthermore, let a cyclic code $\CYCa$ with defining set $\defset{\CYCa}$
be given, such that for the parameters $\HTconstb=1$ and $m=\nuHT+2$ the normalized HT bound of~\refeq{eq_specialHT} holds.
We illustrate the defining set $\defset{\SPCa} = \{0\}$ of $\SPCa$ with length $\LCn=\nuHT+2$ and the defining set $\defset{\CYCa}$ in Table~\ref{tab_HT_ParityCheck}.
\begin{table}[htbp]
\centering
\caption{Defining sets $\defset{\CYCa}$ of a given cyclic code $\CYCa$ and $\defset{\SPCa}$ of its associated single parity check code \SPCa{} of length $\LCn$ in the interval $\left[0,m(\muHT-1)\right]$.}
\label{tab_HT_ParityCheck}
\begin{tabularx}{.95\columnwidth}{l|ZZZZ;{2pt/2pt}ZZZc;{2pt/2pt}ZZc;{2pt/2pt}Z}
\centering
$\defset{\CYCa}$ & $\square$ & 1 & .. & $m$-$1$ & $\square$ & $m$+$1$ & .. & 2$m$-$1$ & $\square$ &
.. & $m$($\muHT$-1)-1 & $\square$\\
$\defset{\SPCa}$  & 0 & $\square$ & .. & $\square$ & 0 & $\square$ & .. & $\square$ & 0 & .. & $\square$ & 0\\
\end{tabularx}
\end{table}
The sequence is illustrated in terms of parameters of the HT bound~\refeq{eq_specialHT}.
For this special case, the non-zero-locator code \LC{\LCq}{\LCn}{\LCk}{\LCd} is a \SPC{\LCn} code. We have:
\begin{equation*}
 \LCn = \nuHT+ 2, \quad \LCk  = \nuHT +1, \quad \LCd  = 2,
 \end{equation*}
and we obtain a zero-sequence of length $\mulo -1 = m(\muHT-1)+1$. From Theorem~\ref{theo_minimumdistance} we obtain:
\begin{align} \label{eq_boundsingeparitycheck}
 d^{\ast} & = \left \lceil \frac{m(\muHT-1)+2}{2} \right\rceil = \left \lceil \frac{(\nuHT+2)\muHT - \nuHT}{2} \right\rceil =  \left \lceil \muHT + \frac{\nuHT (\muHT -1)}{2} \right\rceil,
\end{align}
where we used $m = \nuHT+2$.
In Fig.~\ref{fig_HTLocator1} we illustrate $d^{\ast}$ of~\refeq{eq_boundsingeparitycheck} for different parameters $\nuHT$ and $\muHT$.
\begin{figure*}[htb]
\normalsize
 \centering
\includegraphics[scale=1.4]{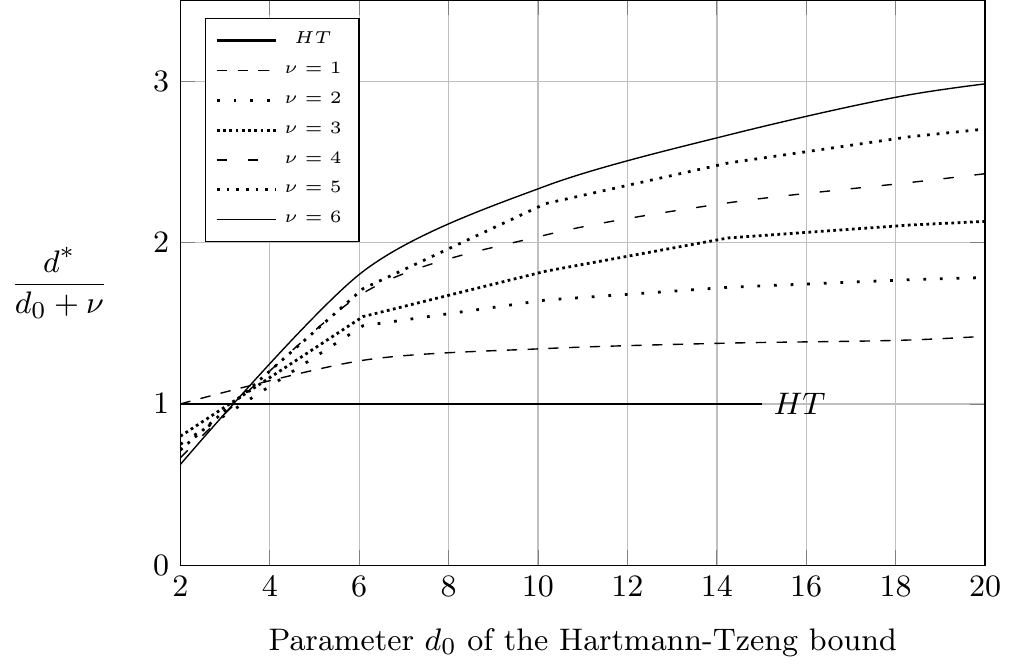}
\caption{Illustration of the fraction $d^{\ast}/(\muHT+\nuHT)$ of our bound $d^{\ast}$ of~\refeq{eq_dfracRS} to the Hartmann--Tzeng bound $\muHT+\nuHT$ for 
$\nuHT = 1,\dots,6$ and $\muHT=2,\dots,20$. The parameters of the HT bound are $m=\nuHT+2$ (see Table~\ref{tab_HT_ParityCheck}). We used a single parity check code as non-zero-locator code. Our bound $d^{\ast}$ is better than the HT bound for $\muHT>3$.}
\label{fig_HTLocator1}
\end{figure*}
For $\muHT \geq 4$ (independently from $\nuHT$) our bound improves the HT bound (see Proposition~\ref{prop_OurBoundBetterThanHT} in the next subsection). Note that for $\nuHT = 0$ the HT bound and our bound coincide with the BCH bound. Let us study the following example.
\begin{example}[Parity Check Code as Non-Zero-Locator Code]
Let us consider the binary reversible~\cite{Massey_ReversibleCodes_1964} cyclic code \CYC{2}{65}{41}{8} with the defining set $D_{\mathcal C} =  \COS{1}{65} \cup \COS{5}{65}$. We know that
\begin{align*}
\{\square,-5,-4,\square,-2,-1,\square,1,2,\square,4,5,\square\} \subseteq D_{\mathcal C}.
\end{align*}
The HT bound gives a lower bound of $d \geq 6$ on the minimum distance of $\CYC{2}{65}{41}{8}$ (for $\HTconstb=-5$, $m=3$, $\muHT=5$ and $\nuHT=1$).
We can associate the  single parity check code \SPCa(3,2,2) over $\F{2^2}$ with generator polynomial $g(x)=x-1$ as a non-zero-locator code for \CYC{2}{65}{41}{8}. The defining sets $\defset{\CYCa}$ and  $\defset{\SPCa}$ are shown in Table~\ref{tab_HT_ParityCheck_Ex1}.
\begin{table}[htbp]
\centering
\caption{Subset of the defining sets $\defset{\CYCa}$ of the \CYC{2}{65}{41}{8} code in the
interval $\left[-6,6\right]$. The set $\defset{\SPCa}$ is the defining set of a single parity check code $\SPCa$ of length $\LCn = 3$ that is the associated non-zero-locator code.}
\label{tab_HT_ParityCheck_Ex1}
\begin{tabularx}{.95\columnwidth}{l|cZZ;{2pt/2pt}ZZZ;{2pt/2pt}ZZZ;{2pt/2pt}ZZZ;{2pt/2pt}Z}
$\defset{\CYCa}$ & $\square$ & -5 & -4 & $\square$ & -2 & -1 & $\square$ & 1 & 2 & $\square$ & 4 & 5 & $\square$ \\ 
$\defset{\SPCa}$ & 0 & $\square$ & $\square$ & 0 & $\square$ & $\square$ & 0 & $\square$ & $\square$ &
0 & $\square$  & $\square$ & 0
\end{tabularx}
\end{table}
With~\refeq{eq_boundsingeparitycheck} we obtain for $d^{\ast}$:
\begin{equation*}
d^{\ast} =  \left \lceil \muHT + \frac{\nuHT (\muHT -1)}{2} \right\rceil = \left \lceil 5 + \frac{1(5-1)}{2} \right\rceil = 7.
\end{equation*}
Furthermore, we can decode up to $(d^{\ast}-1)/2 = 3$ errors for \CYC{2}{65}{41}{8} (see Section~\ref{sec_decoding}).
\end{example}

\subsection{Cyclic Reed--Solomon Codes as Non-Zero-Locator Codes} \label{subsec_RS} 
Let a $q$-ary cyclic code $\CYCa$ with defining set $\defset{\CYCa}$ be given such that for the parameters $\HTconstb=1$ and $m > \nuHT+2$, the normalized Hartmann--Tzeng bound of~\refeq{eq_specialHT} with $\muHT>2$ and $\nuHT>0$ holds.
Let a cyclic Reed--Solomon code \RS{\LCq}{\LCn}{\LCk}{\delta} over an extension field \F{\LCq} of $\F{q}$ with
\begin{equation*}
\LCn = m, \quad \LCk = \nuHT +1, \quad \LCd = m -\nuHT, \quad \delta=0
\end{equation*}
as in Definition~\ref{def_CyclicRS} be the associated non-zero-locator code.

Table~\ref{tab_HT_ReedSolomon} shows the defining set $\defset{\CYCa}$ and the defining set $\defset{\RSa}$ of \RS{\LCq}{m}{\nuHT+1}{0}.

\begin{table*}[htbp]
\centering
\caption{Defining sets $\defset{\CYCa}$ for $\HTconstb=1$ and $m$ of the HT bound~\refeq{eq_specialHT} and $\defset{\RSa}$ of the associated non-zero-locator code in the
interval $\left[-(m-\nuHT)-1,m(\muHT-1)\right]$.}
\label{tab_HT_ReedSolomon}
\begin{tabularx}{.99\columnwidth}{p{0.04\columnwidth}|p{0.01\columnwidth}p{0.01\columnwidth}p{0.07\columnwidth}p{0.01\columnwidth}p{0.01\columnwidth}p{0.035\columnwidth};{2pt/2pt}p{0.01\columnwidth}p{0.01\columnwidth}p{0.07\columnwidth}p{0.03\columnwidth}p{0.01\columnwidth}p{0.08\columnwidth};{2pt/2pt}p{0.01\columnwidth}p{0.01\columnwidth}p{0.07\columnwidth}p{0.01\columnwidth}p{0.07\columnwidth}}
$\defset{\CYCa}$ & $\square$ & .. & $\square$ & 1 & .. & $\nuHT$+1 & $\square$ & .. & $\square$ & $m$+1 & .. & $m$+$\nuHT$+1 & $\square$ & .. & $\square$ & ..  & $\square$ \\ 
$\defset{\RSa}$ & 0 & .. & $m$-$\nuHT$-$2$ & $\square$ & .. & $\square$ & 0 & .. & $m$-$\nuHT$-$2$ & $\square$ & .. & $\square$ & 0 & .. &  $m$-$\nuHT$-$2$ & ..  &  $m$-$\nuHT$-$2$
\end{tabularx}
\end{table*}
The $\LCn -(\nuHT+2) +1 = m-\nuHT-1$ consecutive zeros of the cyclic Reed--Solomon code \RS{\LCq}{m}{\nuHT+1}{0} fill the missing zeros of the given cyclic code $\CYC{q}{n}{k}{d}$.
The obtained ``zero''-sequence has length $\mulo-1  = m(\muHT-1) + m-\nuHT-1$.
Therefore, we obtain from~\refeq{inequality_for_d}:
\begin{align} \label{eq_dfracRS}
d^{\ast} & = \left \lceil \frac{m(\muHT-1)+m-\nuHT}{m-\nuHT}  \right\rceil = \left \lceil \frac{m\muHT-m+m-\nuHT}{m-\nuHT} \right\rceil = \left \lceil \frac{m\muHT-\nuHT}{m-\nuHT} \right\rceil.
\end{align}
Note that for $m=\nuHT+2$ the Reed--Solomon code is a single parity check code and we obtain the result from~\refeq{eq_boundsingeparitycheck}.
\begin{figure*}[htbp]
\normalsize
 \centering
\includegraphics[scale=1.4]{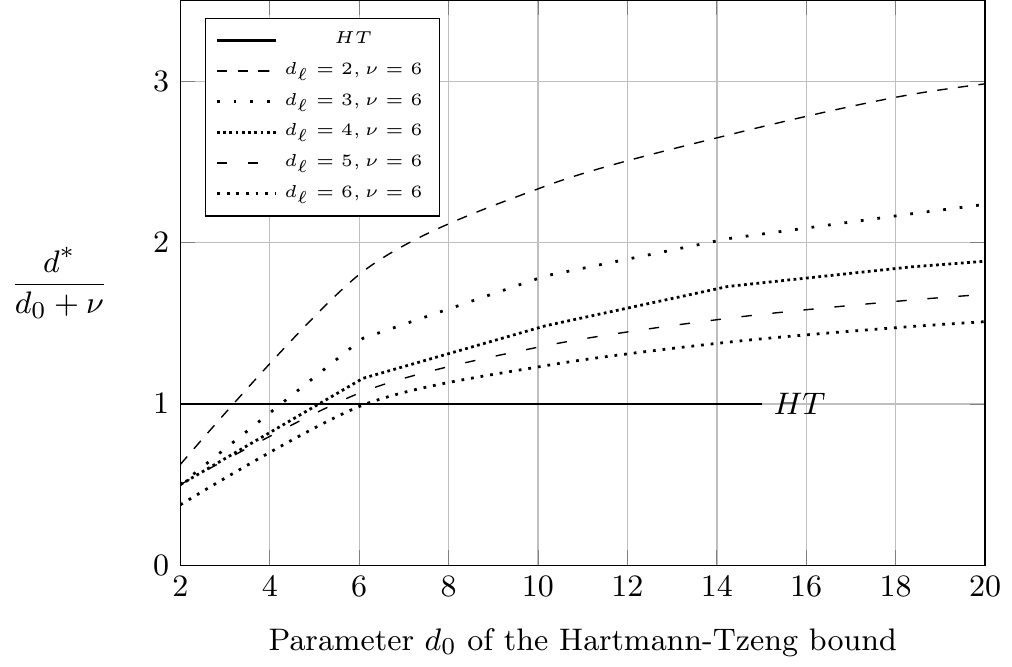}
\caption{Illustration of the fraction $d^{\ast}/(\muHT+\nuHT)$ of our bound $d^{\ast}$ of~\refeq{eq_dfracRS} to the Hartmann--Tzeng bound $\muHT+\nuHT$ for 
$\nuHT = 6$, $\muHT=2,\dots,20$ and $m$.
We used an RS code as non-zero-locator code with minimum distance $\LCd = m-\nuHT$ (see Table~\ref{tab_HT_ReedSolomon}).}
\label{fig_HTLocator2}
\end{figure*}
Let us precise the cases where our bound $d^{\ast}$ is better than the Hartmann--Tzeng bound $\muHT+\nuHT$.
\begin{proposition} \label{prop_OurBoundBetterThanHT}
Let a $q$-ary cyclic code \CYC{q}{n}{k}{d} with a subset of its defining set $\defset{\CYCa}$ with parameters $\HTconstb$, $m$, $\muHT$ and $\nuHT$ as stated in Theorem~\ref{theo_HT} be given. Let $\LC{\LCq}{m}{\nuHT+1}{m-\nuHT} = \RS{\LCq}{m}{\nuHT+1}{0}$ be the associated non-zero-locator code as in Definition~\ref{eq_locatorandcodeword} with $\mu = m(\muHT-1)+m-\nuHT$.
Then for
\begin{equation*} 
\muHT > m-\nuHT+1,
\end{equation*}
$d^{\ast} > \muHT + \nuHT$ holds.
\end{proposition}
\begin{proof}
From \refeq{eq_dfracRS} we have
\begin{equation*}
d^{\ast}  = \left \lceil \frac{m\muHT-\nuHT}{m-\nuHT} \right\rceil 
 = \left \lceil \frac{m\muHT-\muHT\nuHT +\muHT\nuHT - \nuHT}{m-\nuHT} \right\rceil = \left \lceil \muHT + \frac{(\muHT-1)\nuHT}{m-\nuHT} \right\rceil.
\end{equation*}
Obviously, for $d^{\ast} > \muHT + \nuHT$, we require that
\begin{equation*} \label{eq_betterthanHT}
 \frac{(\muHT-1)\nuHT}{m-\nuHT}  > \nuHT \; \Longleftrightarrow \; \muHT  > m-\nuHT+1.
\end{equation*}
\end{proof}
For $m-\nuHT = \LCd = 2$, the associated RS code is a single parity check code and our bound is better than the HT bound for $\muHT > 3$ (see Fig.~\ref{fig_HTLocator1}). 
Some other cases, where the minimum distance of the associated RS code $\LCd=m-\nuHT$ varies between two and six, are illustrated in Fig.~\ref{fig_HTLocator2}.

\subsection{Some Families of Cyclic Codes and Their Connection to Other Bounds}
\label{subsec_familiesOfCode}
We identify some families of cyclic codes and refine our bound on the minimum distance of Theorem~\ref{theo_minimumdistance}. The classification is done by means of the associated non-zero-locator code. For all codes, we can decode up to $\lfloor(d^{\ast}-1)/2 \rfloor$ errors (see Section~\ref{sec_decoding}).

\subsubsection*{Single Parity Check Code as Non-Zero-Locator Code}
Let the defining set $\defset{\CYCa}$ of a given $q$-ary cyclic code \CYC{q}{n}{k}{d} contain the elements as shown in Table~\ref{tab_Class_ParityCheck}. Furthermore, let
$\gcd(n,3)=1$.
\begin{table}[htbp]
\setlength\tabcolsep{5pt}
\centering
\caption{Subset of the defining sets $\defset{\CYCa}$ of a given cyclic 
code $\CYCa$ in the
interval $\left[-10,10\right]$. The set $\defset{\LCa} = \{0\}$ is the defining set of the single parity check code $\LC{2}{3}{2}{2}$.}
\label{tab_Class_ParityCheck}
\begin{tabularx}{.95\columnwidth}{l|c;{2pt/2pt}ZZZ;{2pt/2pt}ZZZ;{2pt/2pt}ZZZ;{2pt/2pt}ZZZ;{2pt/2pt}ZZZ;{2pt/2pt}ZZZ;{2pt/2pt}ZZ}
$\defset{\CYCa}$  & -10 & $\square$ & -8 & -7 & $\square$ & -5 & -4 & $\square$ & -2 & -1 & $\square$ & 1 & 2 & $\square$ & 4 & 5 & $\square$ & 7 & 8 & $\square$ & 10 \\
$\defset{\LCa}$ & $\square$ & 0 & $\square$ & $\square$ & 0 & $\square$ & $\square$ & 0 & $\square$ & $\square$ & 0 & $\square$ & $\square$ & 0 & $\square$  & $\square$ & 0 & $\square$ & $\square$ & 0 & $\square$
\end{tabularx}
\end{table}
We associate a single parity check code $\SPCa(3,2,2)$ and obtain $\mulo=22$ and therefore $d \geq d^{\ast} = 11$.

For binary reversible cyclic codes~\cite{Massey_ReversibleCodes_1964, Zetterberg_CyclicCodesFromIrreduciblePolynomialsForCorrectionOfMultipleErrors_1962} we require only $\{1,5,7\}$ to be a subset of the defining set since the other elements are then included automatically.

If the binary cyclic code is not reversible, the defining set has to contain $\{-7$, $-5$, $-1$,
$1$, $5$, $7\}$. This requirement coincides with the 5-error-correcting pair of~\cite[Proposition 8]{DuursmaKoetter_ErrorLocatingPairs_1994}.  The codes of~\cite[Proposition 7, Example 21 and 22]{DuursmaKoetter_ErrorLocatingPairs_1994} require a smaller subset of their defining set $\defset{\CYCa}$. For these codes, we obtain the same bound on the minimum distance of $\CYCa$.

\subsubsection*{Binary Hamming Code as Non-Zero-Locator Code}
Let the defining set $\defset{\CYCa}$ of a given binary cyclic code \CYC{2}{n}{k}{d} contain the elements as shown in Table~\ref{tab_Class_Hamming}. Furthermore, let
$\gcd(n,7)=1$.

\begin{table*}[htbp]
\centering
\caption{Subset of the defining sets $\defset{\CYCa}$ of a given cyclic 
code $\CYCa$ in the
interval $\left[1,20\right]$. The set $\defset{\LCa}=\{3,5,6\}$ is the defining set of the binary Hamming code $\LC{2}{7}{4}{3}$.}
\label{tab_Class_Hamming}
\begin{tabularx}{.95\columnwidth}{l|ZZZZZZ;{2pt/2pt}ZZZZZZZ;{2pt/2pt}ZZZZZZZ}
$\defset{\CYCa}$ & 1 & 2 & $\square$ & 4 & $\square$ & $\square$ & 7 & 8 & 9 & $\square$ & 11 & $\square$ & $\square$ & 14 & 15 & 16 & $\square$ & 18 & $\square$ & $\square$ \\
$\defset{\LCa}$ & $\square$ & $\square$ & 3 &  $\square$ & 5 & 6 & $\square$ & $\square$ &
$\square$ & 3 & $\square$ & 5 & 6 & $\square$ & $\square$ & $\square$ & 3 & $\square$ & 5 & 6 
\end{tabularx}
\end{table*}
We associate the binary Hamming code $\LC{2}{7}{4}{3}$ with defining set $\defset{\LCa}=\{3,5,6\}$. As shown in Table~\ref{tab_Class_Hamming}, we obtain $\mulo=21$ and therefore $d \geq d^{\ast} = 7$.

For binary cyclic codes we require $\{1,7,9,11,15\}$ to be a subset of the defining set $\defset{\CYCa}$.

\subsubsection*{Reed--Solomon Code as Non-Zero-Locator Code}
Let the defining set $\defset{\CYCa}$ of a given $q$-ary cyclic code \CYC{q}{n}{k}{d} contain the elements as shown in Table~\ref{tab_Class_RSCode}. Furthermore, let
$\gcd(n,4)=1$.
\begin{table*}[htbp]
\centering
\caption{Subset of the defining sets $\defset{\CYCa}$ of a given cyclic 
code $\CYCa$ in the interval $\left[-17,17\right]$ (only odd indexes are illustrated). The set $\defset{\LCa}= \{0,1\}$ is the defining set of a Reed--Solomon code $\RS{\LCq}{4}{2}{0}$. }
\label{tab_Class_RSCode}
\begin{tabularx}{.95\columnwidth}{l|ZZcc;{2pt/2pt}ZZZZ;{2pt/2pt}ZZZZ;{2pt/2pt}ZZZZ;{2pt/2pt}ZZ}
$\defset{\CYCa}$ & $\square$ & $\square$ & -13 & -11 & $\square$ & $\square$ & -5 & -3 & $\square$ & $\square$ & 3 & 5 & $\square$ & $\square$ & 11 & 13 & $\square$ &  $\square$ \\
$\defset{\RSa}$ & 0 & 1 & $\square$ & $\square$ & 0 & 1 & $\square$ & $\square$ & 0 & 1 & $\square$ & $\square$ & 0 & 1 & $\square$ & $\square$ & 0 & 1 
\end{tabularx}
\end{table*}
We associate an RS code $\RS{\LCq}{4}{2}{\delta = 0}$ over $\F{\LCq}$ which is an extension field of $\F{q}$ and consider the sequence
\begin{equation*}
c(\alpha^{-17}) a(\beta^0) , c(\alpha^{-17+2}) a(\beta^{1}) , c(\alpha^{-17+4}) a(\beta^{2}), \dots , c(\alpha^{-17+(\mulo-2)\cdot 2}) a(\beta^{\mulo-2}).
\end{equation*}
We have $\mulo=19$ and with $\LCd=3$, we obtain $d \geq d^{\ast} = 7$.

For binary reversible cyclic codes, we require $\{3,5,11,13\}$ to be a subset of the defining set $\defset{\CYCa}$.

Further families can be found in~\cite{ZehWachterBezza_EfficientDecodingOfSomeClassesArxiv_2011}
and can be seen as special case of this approach. 

\medskip

As previously seen, we identified cyclic codes by means of their potential non-zero-locator codes. To obtain a huge family of cyclic codes, the cardinality of the required subset of their defining set should be small. This implies a high cardinality of the defining set $|\defset{\LCa}|$ of the associated non-zero-locator code \LC{\LCq}{\LCn}{\LCk}{\LCd}. Both leads to a long zero-sequence
\begin{equation*}
c(\alpha^{\LCconst}) a(\beta^0) , c(\alpha^{\LCconst+1}) a(\beta^{1}) , \dots , c(\alpha^{\LCconst+\mulo-2}) a(\beta^{\mulo-2}).
\end{equation*}
On the one hand, we need a low code-rate $\LCk/\LCn$ which implies a high $|\defset{\LCa}|$. On the other hand, the minimum distance $\LCd$ of $\LCa$ should be small to obtain a good bound $d^{\ast}$ according to~\refeq{inequality_for_d}.

This motivates the investigation of small-minimum-distance cyclic codes with lowest code-rate. In a first step, we consider binary cyclic codes with minimum distance two and three.

\section{Binary Cyclic Codes with Minimum Distance Two and Three as Non-Zero-Locator Code} \label{sec_badcyclic}
\subsection{General Idea}
As mentioned in Section~\ref{sec_loccode}, good candidates for non-zero-locator codes are cyclic codes with small minimum distance and lowest code-rate $\LCk/ \LCn $.
We consider binary cyclic codes with minimum distance two and three and lowest code-rate and show their defining set.

Primitive binary cyclic codes with minimum distance three were investigated by Charpin, Tiet\"{a}v\"{a}inen and Zinoviev in~\cite{Charpin_OnBinaryCyclicCodes_1997,Charpin_MinimumDistanceOfNonBinary_1999}. We generalize the results of~\cite{Charpin_OnBinaryCyclicCodes_1997} to binary cyclic codes of arbitrary length and show afterwards the implications, when we want to use them as non-zero-locator codes.
\begin{lemma}\cite{Charpin_OnBinaryCyclicCodes_1997} \label{lem_CyclicDistTwo}
Let $i,j$ with $0 \leq i < j \leq n-1$ be two arbitrary integers that do not belong to the same cyclotomic coset modulo $n$. Then the binary cyclic code \CYC{2}{n}{k}{d} with generator polynomial $g(x) = \COS{i}{n}(x) \cdot \COS{j}{n}(x) $ has minimum distance
two if and only if  $\gcd(n,i,j) > 1$.
\end{lemma}
\begin{proof}
Let $\alpha$ be an $n$th root of unity. A binary cyclic code $\CYCa$ with generator polynomial $g(x) = \COS{i}{n}(x) \cdot \COS{j}{n}(x)$ of length $n$ has minimum distance two if there exist a binomial $c(x) = x^{k} + x^{\ell}$ that fulfills 
\begin{equation*}
c(\alpha^i) = c(\alpha^j) = 0.
\end{equation*}
This holds, if and only if 
\begin{equation*}
 \alpha^{ki} = \alpha^{\ell i} \quad \text{and} \quad \alpha^{kj} = \alpha^{\ell j}
\end{equation*}
or, equivalently,
\begin{equation*}
 (k- \ell )i \equiv (k- \ell )j \equiv 0 \mod n.
\end{equation*}
Both congruences are valid if and only if $n/\gcd(n,i,j)$ divides $k-\ell$. Therefore,
such $k$ and $\ell$ exist if and only if $\gcd(n,i,j) > 1$.
\end{proof}
\begin{theorem}[Binary Cyclic Codes with Minimum Distance Two~\cite{Charpin_OnBinaryCyclicCodes_1997}] \label{theo_CyclicDistTwo}
Let $i_1,i_2,\dots,i_s$ with $0 \leq i_1 < \dots < i_s \leq n-1$ be $s$ arbitrary integers that do not belong to the same cyclotomic coset modulo $n$. Then the binary cyclic code \CYC{2}{n}{k}{d} with generator polynomial 
\begin{equation*}
g(x) = \prod_{j=1}^{s} \COS{i_j}{n}(x)
\end{equation*}
 has minimum distance two if and only if  $\gcd(n,i_1,\dots,i_s) > 1$.
\end{theorem}
We skip the proof of Theorem~\ref{theo_CyclicDistTwo}, because it is straightforward to the proof of Lemma~\ref{lem_CyclicDistTwo}.

The following lemma is a generalization of~\cite[Theorem 1]{Charpin_OnBinaryCyclicCodes_1997} to binary cyclic codes of arbitrary length.
\begin{lemma}[Binary Cyclic Codes with Minimum Distance Three] \label{lem_CyclicDistThree}
Let $i,j$ with $0 \leq i < j \leq n-1$ be arbitrary integers that do not belong to the same cyclotomic coset modulo $n$. Let $g$ be such that $2^g-1$ divides $n$. 
If there exists an integer $r$ with $0 < r < 2^g-1$, where $\gcd(r,2^g-1)=1$, such that 
both $i$ and $j$ are in $\COS{r}{2^g-1}$, then the binary cyclic code \CYC{2}{n}{k}{d} with generator polynomial $g(x) =  \COS{i}{n}(x) \cdot \COS{j}{n}(x)$ has minimum distance $d\leq 3$. If, moreover, $\gcd(n,i,j) = 1$, then $d=3$.
\end{lemma}
\begin{proof}
Let $\gamma$ be a primitive element of $\F{2^s}$, let $z=(2^s-1)/n$ and let $\alpha = \gamma^z$. Let $u=n/(2^g-1)$, then $\beta = \alpha^u = \gamma^{(2^s-1)/(2^g-1)}$, is a primitive element of $\F{2^g}$. Let $b$ be an integer in the interval $[1,2^g-2]$ such that:
\begin{equation*}
1+\beta + \beta^b = 0.
\end{equation*}
Define
\begin{equation*} 
c(x) = 1+ x^{u(1/r)} + x^{u(b/r)},
\end{equation*}
where the quotients $1/r$ and $b/r$ are calculated in the ring $\Z{2^g-1}$ of integers
modulo $2^g-1$.
For $i \in \COS{r}{2^g-1}$, two non-negative integers $k$ and $\ell$ exist such that
\begin{equation*}
i = \ell(2^g-1) + 2^kr.
\end{equation*}
Thus,
\begin{align*}
c(\alpha^i) & = 1+\alpha^{ui(1/r)} + \alpha^{ui(b/r)} \\
 & = 1+ \beta^{i(1/r)} + \beta^{i(b/r)} \\
 & = 1+ \beta^{2^kr(1/r)} + \beta^{2^kr(b/r)} \\
 & = 1+ \beta^{2^k} + \beta^{b 2^k } \\
 & = (1+ \beta + \beta^b)^{2^k} = 0.
\end{align*} 
\end{proof}
Note that in \cite{Charpin_OnBinaryCyclicCodes_1997} the length of the cyclic code was $n=2^s-1$ and $u = (2^s-1)/(2^g-1)$.
\begin{corollary}
Let $\CYCa$ be a binary cyclic code of length $n$. If there exist no $g$, s.t. $(2^g-1) \mid n$, then
$\CYCa$ cannot have minimum distance three.
\end{corollary}
\begin{theorem}[Binary Cyclic Codes with Minimum Distance Three] \label{theo_CyclicDistThree}
Let $i_1,i_2,\dots,i_s$ with $0 \leq i_1 < \dots < i_s \leq n-1$ be $s$ arbitrary integers that do not belong to the same cyclotomic coset modulo $n$. Let $g$ be such that $2^g-1$ divides $n$. 
If there exists an integer $r$ with $0 < r < 2^g-1$, where $\gcd(r,2^g-1)=1$, such that 
all $s$ integers $i_1,i_2,\dots,i_s$ are in $\COS{r}{2^g-1}$, then the binary cyclic code \CYC{2}{n}{k}{d} with generator polynomial 
\begin{equation*}
g(x) =  \prod_{j=1}^s \COS{i_j}{n}(x) 
\end{equation*}
has minimum distance $d\leq 3$. If, moreover, $\gcd(n,i_1,\dots,i_s) = 1$, then $d=3$.
\end{theorem}
We skip the proof of Theorem~\ref{theo_CyclicDistThree}, because it is straightforward to the proof of Lemma~\ref{lem_CyclicDistThree}.

Let us consider a non-primitive binary cyclic code with minimum distance three.
\begin{example}[Non-primitive Binary Cyclic Code with Minimum Distance Three] \label{ex_BinaryCodeDistThree}
Let $n=119 = (2^3-1) \cdot 17$. In this case $g=3$ (see Theorem~\ref{theo_CyclicDistThree}). Then $\{1,11,51\}$ belong to $\COS{1}{7}$ and we have $\gcd(1,11,51)=1$. Therefore the binary cyclic code of length $n=119$ with generator polynomial
\begin{equation*}
 g(x) = \COS{1}{119}(x) \cdot \COS{11}{119}(x) \cdot \COS{51}{119}(x),
\end{equation*}
has dimension $k=68$ and minimum distance $d=3$.
\end{example}

\subsection{Implications for the Non-Zero-Locator Code}
We consider lowest-code-rate binary cyclic codes of minimum distance two and three. They are good candidates for non-zero-locator codes.

We first consider lowest-code-rate binary cyclic codes of minimum distance two.
\begin{proposition}[Lowest-Code-Rate Binary Cyclic Codes With Minimum Distance Two]  \label{theo_badbinaryTwo}
Let $a>1$, $g>1$ and $n$ be three integers, such that $n=ag$. Let $g$ be in the defining set $\defset{\CYCa}$.
Then the binary cyclic code \CYC{2}{n}{k}{2} of length $n$ with defining set:
\begin{equation*}
\defset{\CYCa} = \{0,\square,\dots,\square,g,\square, \dots, \square,2g,\square, \dots,\square, (a-1)g, \square, \dots,\square \} 
\end{equation*}
is the binary cyclic code of smallest dimension $k=a(g-1)$, lowest code-rate $R=(g-1)/g$ and minimum distance two.
\end{proposition}
\begin{proof}
We want to maximize $|\defset{\CYCa}|$ while keeping $d$ of $\CYCa$ at two. Therefore, we select for a given $g$ every cyclotomic coset \COS{i}{n} with $\gcd(i,g)>1$ for all $i=0,\dots,n-1$ to be in $\defset{\CYCa}$ with aimed minimum distance two. One the one hand, this guarantees the maximization of $|\defset{\CYCa}|$ and therefore the minimization of the code-rate. On the other hand, due to the condition $\gcd(i,g)>1$ (Theorem~\ref{theo_CyclicDistTwo}) the minimum distance of $\CYCa$ remains two.
\end{proof}
A direct consequence of Proposition~\ref{theo_badbinaryTwo} is that we do not need to investigate these binary cyclic codes of minimum distance two any more. We obtain the same result when we select a parity check code \SPC{g} as 
non-zero-locator code.

\begin{proposition}[Lowest-Code-Rate Binary Cyclic Codes With Minimum Distance Three] \label{theo_BadCyclicThree}
Let $a>1$, $g>1$ and $n$ be three integers, such that $n=a(2^g-1)$. 
Let $r$ be an integer with $0 < r < 2^g-1$, where $\gcd(r,2^g-1)=1$.
Let $r$ be in the defining set $\defset{\CYCa}$.
Then the binary cyclic code \CYC{2}{n}{k}{3} of length $n$ with defining set:
\begin{equation}
\begin{split}
\defset{\CYCa} = \{ r \cdot i \mod n \mid \, i =  & j(2^{g}-1)+1, j(2^{g}-1)+2, j(2^{g}-1)+4, \dots,\\
& j(2^{g}-1)+2^{g-1} \quad \forall j = 0,\dots,a-1 \}
\end{split}
\end{equation}
is the binary cyclic code with the smallest dimension $k=a(2^g-1-g)$, lowest code-rate $R=(2^g-1-g)/(2^g-1)$ and minimum distance three.
\end{proposition}
\begin{proof}
We want to maximize $|\defset{\CYCa}|$ while keeping $d$ of $\CYCa$ at three.
For a given $r$ and for $(2^g-1) | n$, we select every cyclotomic coset \COS{i}{n} for all $i=0,\dots,n-1$ to be in the $\defset{\CYCa}$ of $\CYCa$ with aimed minimum distance three, such that $i \in \COS{r}{2^g-1}$.
One the one hand, this guarantees the maximization of $|\defset{\CYCa}|$ and therefore the minimization of the code-rate. On the other hand, due to the condition that $\COS{i}{n}$ should be selected such that $i \in \COS{r}{2^g-1}$ (Theorem~\ref{theo_CyclicDistThree}) the minimum distance of $\CYCa$ remains three.
\end{proof}
\begin{remark}
Let $r=1$ in Proposition~\ref{theo_BadCyclicThree}. Then $\COS{1}{2^g-1} = \{1,2,4,\dots,2^{g-1} \}$ is the cyclotomic coset of a binary Hamming code of length $2^g-1$. The defining set of the corresponding lowest-code-rate binary cyclic code is a repetition of the defining set of the Hamming code of length $2^g-1$.
\end{remark}
\begin{example}[Non-primitive Binary Cyclic Code with Minimum Distance Three and Lowest Code-Rate]
Let us again consider Example~\ref{ex_BinaryCodeDistThree} with $n=119 = (2^3-1) \cdot 17$ and $k=68$. The binary cyclic code of length $n=119$ with generator polynomial $g(x) = \COS{1}{119}(x) \cdot \COS{11}{119}(x) \cdot \COS{51}{119}(x)$  and with minimum distance three has lowest code-rate $R=(2^3-1-3)/(2^3-1)=68/119$.
Its defining set $\defset{\CYCa}$ is:
\begin{align*}
\defset{\CYCa}  =  \{ \square,1,2,\square,4,\square,\square,\; \brokenvert \; \square, 8,9,\square,11,\square,\square,\; \brokenvert \; \square, 15,16,\square,18,\square,\square,\; \brokenvert \;  \square,  22,\dots, 116, \square,\square\}.
\end{align*}
\end{example}
A consequence of Proposition~\ref{theo_BadCyclicThree} is that we do not need to investigate any binary cyclic code of minimum distance three any more. We obtain the same result when we take a primitive binary cyclic code with minimum distance three as non-zero-locator code.

\section{Syndrome-Based Decoding of up to $\lfloor (d^{\ast}-1)/2 \rfloor$ Errors} \label{sec_decoding}
\subsection{Syndrome Definition}
Let a $q$-ary cyclic code \CYC{q}{n}{k}{d} and its associated $\LCq$-ary non-zero-locator code \LC{\LCq}{\LCn}{\LCk}{\LCd} with $\gcd(n,\LCn) = 1$ and the integers $\mulo$ and $\LCconst$ be given as in Definition~\ref{def_locatorcode}.  Let $\F{\LCq} = \F{q^\LCextensionorder}$ be an extension field of \F{q}. Let $\alpha \in \F{q^s}$ be a primitive $n$th and let $\beta \in \F{\LCq^{\LCs}}$ be a primitive $\LCn$th root of unity.
Let $r$ denote the least common multiple of $s$ and $\LCextensionorder \cdot \LCs$.
Let $a(x) = \sum_{i \in \mathcal{Z}} a_i x^i$ be a codeword of $\LCa$ of weight $|\mathcal{Z}| = \LCd$.

Let the set $\mathcal E = \{i_0,i_1,\dots,i_{t-1}\}$ with cardinality $|\mathcal E|=t$ be the set of error positions. The corresponding error polynomial is denoted by $e(x) = \sum_{i \in \mathcal E} e_i x^i$. Let the received polynomial be $r(x) = \sum_{i=0}^{n-1} r_i x^i = e(x) + c(x)$.

We define a syndrome polynomial $S(x) \in \Fx{q^r}$ as follows:
\begin{equation} \label{eq_def_synd}
S(x) \defequiv \sum \limits_{j=0}^{\infty} r(\alpha^{j+\LCconst}) a(\beta^j) x^j \mod x^{\mulo-1}.
\end{equation}
Thus, the coefficients $S_j \in \F{q^r}$ of the above defined syndrome polynomial $S(x) = \sum_{j=0}^{\mulo-2} S_j x^j$ are given by
\begin{equation*} \label{eq_synd_explicit}
S_j = \sum_{i=0}^{n-1} r_i \alpha^{i(j+\LCconst)} \cdot \sum_{h=0}^{\LCn-1} a_{h} \beta^{h j},  \quad \forall j=0,\dots,\mulo-2.
\end{equation*}
From Definition~\ref{def_locatorcode} we know that the syndrome polynomial $S(x)$ of~\refeq{eq_def_synd} is independent of the codeword $c(x)$. Now, we can do the same reformulation of the syndrome expression as we did in Section~\ref{sec_loccode} for the codeword $c(x)$ and $a(x)$. We have from~\refeq{eq_def_synd}:
\begin{align*} 
\sum \limits_{j=0}^{\infty} r(\alpha^{j+\LCconst}) a(\beta^j) x^j & \equiv \sum \limits_{j=0}^{\infty} e(\alpha^{j+\LCconst}) a(\beta^j) x^j \mod x^{\mulo-1} \\
 & \equiv \sum \limits_{j=0}^{\infty} \sum_{i \in \mathcal E}  e_i \alpha^{i(j+\LCconst)} a(\beta^j) x^j  \mod x^{\mulo-1},
 \end{align*}
and with~\refeq{eq_PowerSeriesCyclicCode} for $a(x) = \sum_{i \in \mathcal Z} a_i x^i $ we can write:
\begin{align*} 
S(x) & \equiv \sum_{i \in \mathcal E}  e_i \alpha^{i\LCconst} \sum\limits_{j \in \mathcal Z} \frac{a_j}{1-x\alpha^i \beta^j} \mod x^{\mulo-1} \\
 & \equiv \sum_{i \in \mathcal E}  e_i \alpha^{i\LCconst} \frac{\sum\limits_{j \in \mathcal{Z}}  \Big(a_j \prod \limits_{\substack{\ell \in \mathcal{Z}\\ \ell \neq j}} (1-x\alpha^i \beta^{\ell}) \Big) }{\prod \limits_{j \in \mathcal{Z}} \big(1-x \alpha^{i} \beta^j \big)} \mod x^{\mulo-1}.
\end{align*}
Finally, we can write for $S(x)$:
\begin{align} \label{eq_synd_explicitb}
S(x) & \equiv \frac{\sum\limits_{i \in \mathcal E} \Big( e_i \alpha^{i\LCconst} \sum\limits_{j \in \mathcal{Z}} \Big(a_j \prod \limits_{\substack{\ell \in \mathcal{Z}\\ \ell \neq j}} (1-x\alpha^i \beta^{\ell}) \Big)
\prod \limits_{\substack{m \in \mathcal{E}\\ m \neq i}} 
\prod \limits_{s \in \mathcal{Z}} (1-x\alpha^m \beta^s) \Big)}{\prod \limits_{i \in \mathcal{E}} \Big( \prod \limits_{j \in \mathcal{Z}} \big(1-x \alpha^{i} \beta^j \big) \Big)} \mod x^{\mulo-1}.
\end{align}
We use this explicit syndrome representation in the next section, where we define an error-locator and an error-evaluator polynomial.

\subsection{Key Equation}
To simplify the notation, let the two polynomials $f(x)$ and $h(x) \in \Fx{q^r}$ be defined as follows:
\begin{align} 
f(x) & \defeq \prod_{j \in \mathcal{Z}} \big(1-x \beta^j \big), \label{eq_def_loc_poly}  \\
h(x) & \defeq \sum_{j \in \mathcal{Z}} \big( a_j \prod_{\substack{\ell \in \mathcal{Z} \\ \ell \neq j}} (1-x \beta^{\ell}) \big) \label{eq_def_eval_poly}.
\end{align}
Due to $\gcd(n,\LCn)=1$ we have $\gcd(f(x\alpha^i),f(x\alpha^j)) = 1, \, \forall i \neq j$ (for the proof, see Lemma~\ref{lemma_appendixcoprimality} in the Appendix) and therefore each of the $n$ polynomials $f(x\alpha^0),$ $f(x\alpha^1),$ $\dots,$ $f(x\alpha^{n-1})$ can be identified by one root.
Let $\kappa \in \mathcal{Z}$. Then, we have $f(\beta^{-\kappa}) = 0$.
Furthermore, let $n$ distinct roots $\gamma_0,\gamma_1,\dots,\gamma_{n-1}$ be defined as:
\begin{equation} \label{eq_rootdefinition}
\gamma_i \defeq \beta^{-\kappa} \alpha^{-i}, \quad i=0,\dots,n-1.
\end{equation}
Then, each $\gamma_i$ is a root of $f(x\alpha^i)$. Note that each polynomial $f(x\alpha^i)$ has $|\mathcal{Z}| = \LCd$ roots, but we need only one of them.

Now, we can define an error-locator polynomial $\Lambda(x) \in \Fx{q^r}$ as:
\begin{equation} \label{eq_ELP}
\Lambda(x) \defeq  \prod_{i \in \mathcal{E}} f(x\alpha^i).
\end{equation}
The roots $\gamma_i$ of $\Lambda(x)$ from~\refeq{eq_rootdefinition} tell us where the errors are.
The corresponding error-evaluator polynomial $\Omega(x) \in \Fx{q^r}$ is defined as:
\begin{equation} \label{eq_EEP}
\Omega(x) \defeq \sum\limits_{i \in \mathcal E } \Big(e_i \alpha^{i\LCconst} h(x\alpha^i) \prod_{\substack{\ell \in \mathcal E\\ \ell  \neq i}} f(x \alpha^{\ell}) \Big).
\end{equation}
We relate the syndrome definition of~\refeq{eq_synd_explicitb}, the error-locator polynomial $\Lambda(x)$ of~\refeq{eq_ELP} and the error-evaluator polynomial $\Omega(x)$ of~\refeq{eq_EEP} in form of a \textit{Key Equation}:
\begin{equation} \label{eq_KeyEquation}
 \begin{split}
 S(x) & \equiv \frac{\Omega(x)}{\Lambda(x)}  \mod x^{\mulo-1}, \; \text{with} \\
\deg \Lambda(x)  = t \cdot \LCd, \quad & \deg \Omega(x) \leq t \cdot \LCd-1 < \deg \Lambda(x).
  \end{split}
\end{equation}
Solving~\refeq{eq_KeyEquation} is similar to the decoding of~\cite{Sugiyama_AMethodOfSolving_1975} and we will not go into details. The Extended Euclidean Algorithm (EEA,\cite{Sugiyama_AMethodOfSolving_1975}) with input polynomial $S(x)$ as defined in~\refeq{eq_def_synd} and the monomial $x^{\mulo-1}$ and an adapted stopping rule can be used to solve~\refeq{eq_KeyEquation} and we obtain $\Lambda(x)$ and $\Omega(x)$.

\subsection{Error Evaluation: A Generalized Forney's Formula}
To determine the $t$ error values  $e_{i_0},e_{i_1},\dots,e_{i_{t-1}}$ from the error-locator polynomial $\Lambda(x)$ and error-locator polynomial $\Omega(x)$, we develop an explicit expression of the error-values (like Forney's formula~\cite{Forney_OndecodingBCHcodes_1965}) in the following.
\begin{proposition}(Error Evaluation)
Let a $q$-ary cyclic code \CYC{q}{n}{k}{d} and its associated non-zero-locator code \LC{\LCq}{\LCn}{\LCk}{\LCd} with $\gcd(n,\LCn) = 1$ and the integers $\mulo$ and $\LCconst$ be given as in Definition~\ref{def_locatorcode}. Let $\alpha \in \F{q^s}$ be a primitive $n$th and let $\beta \in \F{q^{\LCextensionorder \cdot \LCs}}$ be a primitive $\LCn$th root of unity. Let $r$ be the least common multiple of $s$ and $\LCextensionorder \cdot \LCs$.

Furthermore, let $\gamma_0,\gamma_1,\dots,\gamma_{n-1}$ be given as in~\refeq{eq_rootdefinition} and let two polynomials $\Lambda(x)$ and $\Omega(x) \in \Fx{q^{r}}$ be given as in~\refeq{eq_ELP} and~\refeq{eq_EEP}. Then the error values $e_i$ for all $i \in \mathcal E$ are:
\begin{align}
e_i &  = \frac{\Omega(\gamma_i)}{\alpha^{i\LCconst} \cdot h(\gamma_i \alpha^i) \cdot \prod \limits_{\substack{\ell \in \mathcal E\\ \ell \neq i}} f(\gamma_i \alpha^{\ell}) } \nonumber \\
 & = \frac{\Omega(\gamma_i) \cdot f'(\gamma_i \alpha^i) }{\Lambda'(\gamma_i) \cdot \alpha^{i\LCconst} \cdot h(\gamma_i \alpha^i) }.
\end{align}
\end{proposition}
\begin{proof}
The error-evaluator polynomial $\Omega(x)$ of~\refeq{eq_EEP} evaluated at $\gamma_i$ is explicitly
\begin{equation*}
\Omega(\gamma_i)  = e_i \cdot \alpha^{i\LCconst} \cdot h(\gamma_i\alpha^i) \prod_{\substack{\ell \in \mathcal E\\ \ell \neq i}} f(\gamma_i \alpha^{\ell}).
\end{equation*}
The derivative $\Lambda'(x)$ of the error-locator polynomial is
\begin{equation*}
\Lambda'(x) = \sum_{i \in \mathcal E} \big( f'(x \alpha^i) \prod_{\substack{\ell \in \mathcal E\\ \ell \neq i}} f(x\alpha^{\ell}) \big).
\end{equation*}
Its evaluation at $\gamma_i$ simplifies to
\begin{align*}
\Lambda'(\gamma_i) & = f'(\gamma_i \alpha^i ) \prod_{\substack{\ell \in \mathcal E\\ \ell \neq i}} f(\gamma_i \alpha^{\ell} ).
\end{align*}
\end{proof}
Note that the classical decoding up to the half the BCH bound of a cyclic code $\CYCa$ corresponds to the case where the associated non-zero-locator code \LCa{} is the set of all vectors of length $\LCn=\LCk$ over $\F{\LCq}$. The zero-sequence of length $\mulo-1$ is the longest set of consecutive zeros of $\CYCa$. Then we can choose $a(x)=1$ and we obtain the classical syndrome definition, key equation and Forney's formula.

\section{Conclusion and Outlook} \label{sec_conclusion}
We presented a new technique that uses low-rate cyclic codes with small minimum distances --- so-called non-zero-locator codes --- to bound the minimum distance of $q$-ary cyclic codes.
The algebraic description gives a generalized Key Equation and allows an efficient decoding.
We derived some properties of binary cyclic codes of minimum distance two and three and lowest code-rate.

Future work is to find lowest-code-rate small-minimum-distance non-binary cyclic codes and relate them to our method and bound the minimum distance of other cyclic codes. Combined error-erasure decoding with our proposed method seems to be possible. 

\subsubsection*{Acknowledgments}
We thank the anonymous referees for valuable comments that improved the presentation of this paper.

The authors wish to thank Antonia Wachter-Zeh and Daniel Augot for fruitful discussions.
This work has been supported by German Research Council ``Deutsche Forschungsgemeinschaft'' (DFG) under grant BO~867/22-1.

\section*{Appendix} 
\begin{lemma}[Coprimality of $n$ and $\LCn$] \label{lemma_appendixcoprimality}
Let $[n]$ denote the set of integers $\{ 0,1,\dots,n-1 \}$ and let $\mathcal{Z}$ be a subset of $[\LCn]$.
Let $\alpha$ be an element of order $n$ in $\F{q^s}$ and let $\beta$ denote a primitive element of order $\LCn$ in $\F{\LCq^{\LCs}}$, where $\F{\LCq} = \F{q^\LCextensionorder}$. Let $r$ denote the least common multiple of $s$ and $\LCextensionorder \cdot \LCs$ and let $\gamma$ be a primitive
element in $\F{q^{r}}$. Let $N = q^r-1$. Then $\alpha = \gamma^{N/n}$ and $ \beta = \gamma^{N/\LCn}$.
We consider univariate polynomials in $\F{q^{r}}[x]$. If $\gcd(n,\LCn)=1$ then 
\begin{equation} \label{eq_gcd_codenzl}
 \gcd \Big( \prod_{m \in \mathcal{Z}}  (1-x\alpha^i\beta^m), \prod_{m \in \mathcal{Z}} (1-x\alpha^j\beta^m) \Big) = 1 
\end{equation}
holds $\forall i,j \in [n]$ with $i \neq j$.
\end{lemma}
\begin{proof}
We show that the contrary does not hold. If $\refeq{eq_gcd_codenzl}$ does not hold, then 
there exist a $i$ and $j$ with $i > j$ and $m,m' \in \mathcal{Z} $  with $m \neq m'$ such that
\begin{align}
\alpha^i \beta^m & = \alpha^j \beta^{m'} \nonumber \\
\alpha^{i-j} & = \beta^{m'-m} \label{eq_app_eq1}
\end{align}
holds.
Let us express \refeq{eq_app_eq1} in terms of $\gamma$. We obtain:
\begin{align*}
\gamma^{\frac{N}{n}(i-j)} & = \gamma^{\frac{N}{\LCn}(m'-m)} \\
\gamma^{\frac{N}{n \cdot \LCn} \big((i-j)\LCn - (m'-m) n \big)} & = 1 \\
\Rightarrow (i-j)\LCn -(m'-m) n & = \lambda \cdot n \cdot \LCn.
\end{align*}
We know that $i-j$ is smaller than $n$ and $m'-m$ is smaller than $\LCn$.
This implies that $\lambda$ is zero.
We have:
\begin{align*}
(i-j)\LCn & = (m'-m)n \\
\Rightarrow \LCn & | (m'-m)n 
\end{align*}
But $(m'-m)<\LCn$ and this implies that $\gcd(n,\LCn) \neq 1$.
\end{proof}


\end{document}